\RequirePackage[l2tabu, orthodox]{nag}		
\documentclass[10pt]{article}

\usepackage[T1]{fontenc}
\usepackage{lmodern}
\usepackage{natbib}
\usepackage{microtype}
\usepackage[frenchmath]{mathastext}					
\usepackage{amsmath}                        
\numberwithin{equation}{section}
\usepackage{amsfonts}
\usepackage{graphicx}                      
\usepackage[plainpages=false, pdfpagelabels]{hyperref} 
	\hypersetup{
		colorlinks   = true,
		citecolor    = RoyalBlue, 
		linkcolor    = RubineRed, 
		urlcolor     = Turquoise
	}
\usepackage[dvipsnames]{xcolor}
\usepackage[all]{xy}
\usepackage{amssymb}                        
\usepackage[mathscr]{eucal}                 
\usepackage{dsfont}													
\usepackage[paperwidth=8.5in,paperheight=11.00in,top=1.25in, bottom=1.25in, left=1.00in, right=1.00in]{geometry}
\usepackage{mathtools}                      
\mathtoolsset{showonlyrefs=true}            
\usepackage{fixltx2e,amsmath}               
\MakeRobust{\eqref}
\usepackage{microtype}											
\usepackage{amsthm}                         
\allowdisplaybreaks                         
\theoremstyle{plain}
\newtheorem{theorem}{Theorem}
\numberwithin{theorem}{section}
\newtheorem{proposition}[theorem]{Proposition}

\theoremstyle{definition}
\newtheorem{definition}[theorem]{Definition}

\newcommand{\<}{\langle}
\renewcommand{\>}{\rangle}
\renewcommand{\(}{\left(}
\renewcommand{\)}{\right)}
\renewcommand{\[}{\left[}
\renewcommand{\]}{\right]}

\newcommand\Cb{\mathds{C}}
\newcommand\Eb{\mathds{E}}
\newcommand\Fb{\mathds{F}}

\newcommand\Pb{\mathds{P}}

\newcommand\Rb{\mathds{R}}

\newcommand\Fc{\mathscr{F}}

\newcommand\Nc{\mathscr{N}}
\newcommand\Oc{\mathscr{O}}

\newcommand\eps{\varepsilon}

\newcommand\Om{\Omega}
\newcommand\sig{\sigma}

\newcommand\lam{\lambda}
\newcommand\del{\delta}

\newcommand\rhob{\bar{\rho}}

\renewcommand\d{\partial}
\newcommand{\ind}{\perp \! \! \! \perp}
\newcommand\ii{\mathtt{i}}
\newcommand\dd{\mathrm{d}}
\newcommand\ee{\mathrm{e}}

\renewcommand\Re{\textup{Re}\,}
\renewcommand\Im{\textup{Im}\,}

\begin{document}

\title{On Carr and Lee's correlation immunization strategy}

\author{
Jimin Lin
\thanks{
Department of Applied Mathematics, University of Washington.
\textbf{e-mail}: \url{jmlin@uw.edu }}
\and
Matthew Lorig
\thanks{
Department of Applied Mathematics, University of Washington.
\textbf{e-mail}: \url{mlorig@uw.edu}}
}

\date{This version: \today}

\maketitle

\begin{abstract}
In their seminal work \cite{rrvd} show how to robustly price and replicate a variety of claims written on the quadratic variation of a risky asset under the assumption that the asset's volatility process is independent of the Brownian motion that drives the asset's price.  Additionally, they propose a correlation immunization strategy that minimizes the pricing and hedging error that results when the correlation between the risky asset's price and volatility is nonzero.  In this paper,
{we show that the correlation immunization strategy is the only strategy among the class of strategies discussed in \cite{rrvd} that results in real-valued hedging portfolios when the correlation between the asset's price and volatility is nonzero}.
Additionally, we perform a number of Monte Carlo experiments to test the effectiveness of Carr and Lee's immunization strategy.  Our results indicate that the correlation immunization method is an effective means of reducing pricing and hedging errors that result from nonzero correlation.
\end{abstract}


%
%

\section{Introduction}
\label{sec:intro}
\textit{Volatility} is a catch-all phrase used by practitioners and academics to quantify the uncertainty of a risky asset's value.  Common measures of volatility include Black-Scholes implied volatility, Bachelier implied volatility, instantaneous volatility, and realized quadratic variation (also called ``realized variance'').  A \textit{volatility derivative} is any derivative asset whose payoff depends on some measure of volatility.  For the purpose of this paper, we will focus on volatility derivatives whose payoffs are functions of the realized quadratic variation of the $\log$ price of a risky asset.    Common examples of volatility derivatives of this sort include variance swaps, volatility swaps, and puts and calls on realized variance.  Volatility derivatives play a few important roles.  First, they can be used as hedging instruments for European options.  Second, as instantaneous volatility is known to be negatively correlated with price (in equity markets), a long position in a volatility or variance swap can be used to protect a portfolio's value in the case of a market crash.

Like all derivative assets, volatility derivatives could in principle be priced by choosing a parametric model for the underlying risky asset and computing risk-neutral expectations of payoffs either analytically (if possible) or by Monte Carlo simulation.  However, this parametric approach leads to a great deal of model misspecification risk.  An alternative nonparametric approach, which has enjoyed great success, is to assume only that the underlying risky asset has continuous sample paths and attempt to price volatility derivatives  \textit{relative} to the value of (liquidly traded and efficiently priced) European calls and puts.  The first step in this direction was taken by \cite{neuberger} and \cite{dupire1993model}, who showed independently that the fair strike of a variance swap has the same value as a European $\log$ contract on the underlying risky asset.  They further showed that the floating leg of a variance swap could be replicated by holding a European $\log$ contract and keeping a fixed dollar amount in the underlying risky asset.  Although $\log$ contracts do not trade, they can in theory be synthesized from a continuous strip of calls and puts, as described in \cite{carrmadan1998}.

Another significant step in the nonparametric valuation of volatility derivatives was taken by \cite{rrvd}, who showed how to price and replicate  a large class of nonlinear payoffs of realized variance under the additional assumption that volatility process of the underlying risky asset evolves independently of the Brownian motion that drives the asset's price.

{
A number of papers have built upon the methodology of \cite{rrvd}, whose work was first made available as a discussion paper in 2005.
For example, \cite{zhu-thesis} uses Carr and Lee's methodology to price options on levered exchange traded funds relative to European options on the underlying.
\cite{bscpv} provides pricing and hedging strategies for hybrid barrier-style claims on price and volatility.
\cite{friz2005valuation} perform a detailed mathematical analysis of volatility swaps and calls on variance and show that the latter leads to
an ill-posed problem that can be solved using regularization techniques.
\cite{tvo} robustly price and replicate so-called ``target volatility'' options.
\cite{vs-tclp} and \cite{vs-tcmp} show how to robustly price variance swaps in the presence of both stochastic volatility and jumps.}

{
The work of Carr and Lee has highlighted the need for robust pricing and hedging methods for all sorts of path-dependent options -- not only those related to realized variance.  For example, \cite{forde2010robust} give robust approximations for the prices of arithmetic Asian options in the presence of stochastic volatility.
\cite{spx-vix} derives a model-free link between options on the SPX index and options on the VIX.
And \cite{pcs} provide pricing and replication strategies for a variety of barrier-style claims.
}

Fully aware that instantaneous volatility is empirically negatively correlated with price (the leverage effect), \cite{rrvd} developed a \textit{correlation immunization} strategy, and showed formally that the pricing and hedging error associated with this strategy was on the order of correlation \textit{squared} (i.e., they succeeded in eliminating the first order effects of correlation).
{
In this paper, we show that, among the pricing and hedging strategies discussed in \cite{rrvd}, the correlation immunization strategy is the only strategy that results in real-valued hedging portfolios when the correlation between the risky asset's price and volatility is nonzero.  This is clearly an important consideration for the practical implementation of Carr and Lee's pricing and hedging methodology, as assets in the real world are not complex-valued.
}

Assuming one uses the correlation immunization strategy, the asymptotic result of Carr and Lee tells us how fast the pricing and hedging errors go to zero as correlation goes to zero, but it does not tell us, for a fixed correlation, how large the pricing and hedging errors may be.  {One of the purposes} of this paper is to carry out a numerical investigation of the pricing and hedging errors associated with the correlation immunization strategy for fixed values of correlation by performing a series of Monte Carlo tests.

The rest of this paper proceeds as follows:
In Section \ref{sec:model} we introduce a nonparametric model for a risky asset.
In Section \ref{sec:review} we review the main results from \cite{rrvd} {and prove that the correlation immunization strategy always results in real-valued hedging portfolios}.
And in Section \ref{sec:monte-carlo}, we present the results of our Monte Carlo simulations.
Lastly, in Section \ref{sec:conclusion} we offer some closing remarks.


%
%

\section{Market model}
\label{sec:model}
Throughout this paper, we work in the setting of \cite{rrvd}.  Specifically, we consider a frictionless market (i.e., no transaction costs) and fix an arbitrary but finite time horizon $T<\infty$.  For simplicity, we assume zero interest rates, no arbitrage, and take as given an equivalent martingale measure (EMM) $\Pb$ chosen by the market on a complete filtered probability space $(\Om,\Fc,\Fb,\Pb)$.  The filtration $\Fb=(\Fc_t)_{0 \leq t \leq T}$ represents the history of the market.

Let $B = (B_t)_{0 \leq t \leq T}$ represent the value of a zero-coupon bond maturing at time $T$.  As the risk-free rate of interest is zero by assumption, we have $B_t = 1$ for all $t \in [0,T]$.  Let $S = (S_t)_{0 \leq t \leq T}$ represent the value of a risky asset.  We assume $S$ is strictly positive and has continuous sample paths.  To rule out arbitrage, the price of the asset $S$ must be a martingale under the pricing measure $\Pb$.  As such, there exists a non-negative, $\Fb$-adapted stochastic process $\sig = (\sig_t)_{0 \leq t \leq T}$, called the \textit{volatility process} such that
\begin{align}
\dd S_t
	&=	\sig_t S_t \dd W_t , &
S_0
	&>	0 , 
\end{align} 
where $W$ is a $(\Pb,\Fb)$-Brownian motion.  Without loss of generality, we may decompose $W$ as follows
\begin{align}
W
	&=	\rhob W^1 + \rho W^2 , &
\rhob
	&:=	\sqrt{1-\rho^2} , &
|\rho|
	&\leq 1 . \label{eq:W-decomp}
\end{align}
where $W^1$ and $W^2$ are independent $(\Pb,\Fb)$-Brownian motions and where the volatility process $\sig$ and the Brownian motion $W^1$ are independent (i.e., $\sig \ind W^1$).  We shall refer to the parameter $\rho$ as the \textit{correlation}.  Note that when $\rho = 0$ we have $W = W^1$ and hence $\sig \ind W$.

It will be convenient to introduce $X = (X_t)_{0 \leq t \leq T}$, the $\log$ price process
\begin{align}
X_t
	&=	\log S_t .
\end{align}
As $S$ is strictly positive by assumption, the process $X$ is well-defined and finite for all $t \in [0,T]$.  A simple application of It\^o's Lemma yields
\begin{align}
\dd X_t
	&=	-\tfrac{1}{2} \sig_t^2 \dd t + \sig_t \dd W_t , &
X_0
	&=	\log S_0 . \label{eq:dX}
\end{align}
In this paper, we shall be concerned with path-dependent claims with payoffs at time $T$ of the form $\varphi(\<X\>_T)$, where $\<X\>$ denotes the quadratic variation process of $X$.  Note that
\begin{align}
\<X\>_T
	&=	\int_0^T \sig_t^2 \dd t . \label{eq:qv}
\end{align}
Let $V = (V_t)_{t \leq T}$ be the value of a claim with payoff $\varphi(\<X\>_T)$.  Under the assumption of no arbitrage and zero interest rates, we have
\begin{align}
V_t
	&=	\Eb_t \varphi(\<X\>_T) ,
\end{align}
where $\Eb_{t} \, \cdot \, := \Eb[ \, \cdot \, | \Fc_{t}]$ denotes the $\Fc_t$-conditional expectation under $\Pb$.

We assume that a European call or put with maturity $T$ trades at every strike $K \in (0,\infty)$.
As \cite{carrmadan1998} show, if $g : \Rb \to \Rb$ is a difference of convex functions, then the $T$-maturity European claim with payoff $g(X_T)$ can be perfectly replicated with a static portfolio of bonds $B$, shares of the underlying $S$ and a basket of calls and puts.  Thus, we may (and do) treat all $T$-maturity European claims on $X$ as traded assets.  The price of a $T$-maturity European claim with payoff $g(X_T)$ is equal to the value of the static replicating portfolio and is therefore observable.

%
%

\section{Main results from \cite{rrvd}}
\label{sec:review}

In this section, we briefly review the main results from \cite{rrvd}.  First, in Section \ref{sec:ind}, we present exact pricing and replications results for the case of zero correlation.  Then, in Section \ref{sec:corr}, we present Carr and Lee's correlation immunization strategy for approximate pricing and replication in the case of nonzero correlation.  We also present what we believe is a new result (Proposition \ref{thm:real}), which establishes that, when the volatility derivative payoff is real, so is the associated correlation immunized hedging strategy.


\subsection{Exact pricing and replication under zero correlation}
\label{sec:ind}


In what follows, we shall consider claims with $\Cb$-valued payoffs.  The pricing and hedging results we present should be understood to hold for the real and imaginary parts separately.  We begin with a proposition that relates the price of an exponential claim on realized quadratic variation $\<X\>_T$ to the price of an exponential claim on $\log$ price $X_T$ when $\rho = 0$.

\begin{proposition}[Pricing of Exponential Claims]
\label{thm:pricing}
Assume $\rho = 0$.  Define a function $u^\pm: \Cb \to \Cb$ as follows
\begin{align}
u^\pm(s)
	&=	\ii \( - \tfrac{1}{2} \pm \sqrt{\tfrac{1}{4} + 2 \ii s} \) . \label{eq:u}
\end{align}
For any $s \in \Cb$, define processes $N^\pm(s)=(N_t^\pm(s))_{0 \leq t \leq T}$ and $Q^\pm(s)=(Q_t^\pm(s))_{0 \leq t \leq T}$ as follows
\begin{align}
N_t^\pm(s)
	&:= \ee^{ -\ii u^\pm(s) X_t + \ii s \<X\>_t} , &
Q_t^\pm(s)
	&:= \Eb_t \ee^{ \ii u^\pm(s) X_T} . \label{eq:NQ}
\end{align}
Then, for any $t \leq T$, we have
\begin{align}
\Eb_t \ee^{ \ii s \<X\>_T} 
	&=	N_t^\pm(s) Q_t^\pm(s) . \label{eq:E1=E2}
\end{align}
\end{proposition}

\begin{proof}
See \cite[Proposition 5.1]{rrvd}. 
\end{proof}


Observe that the left-hand side of \eqref{eq:E1=E2} is the time $t$ price of a path-dependent claim with payoff $\ee^{ \ii s \<X\>_T}$ and the right-hand side is the product of an $\Fc_t$-measurable quantity $N_t^\pm(s)$ 
and the time-$t$ value of a European claim $Q_t^\pm(s)$.
As the time $t$-value of the European claim $Q_t^\pm(s)$ can be deduced from call and put prices, equation \eqref{eq:E1=E2} can be viewed as a pricing formula for exponential claims on realized quadratic variation $\<X\>_T$.
We now turn our attention to the replication of such claims.

\begin{proposition}[Replication of Exponential Claims]
\label{thm:replication}
Assume $\rho = 0$.
Let the function $u^\pm$ be as given in \eqref{eq:u} and let the processes $N^\pm(s)$ and $Q^\pm(s)$ be as given in \eqref{eq:NQ}. 
For any $s \in \Cb$, define a self-financing portfolio whose value $\Pi^\pm(s) = (\Pi_t^\pm(s))_{0 \leq t \leq T}$ is given by
\begin{align}
\dd \Pi_t^\pm(s)
	&=	 N_t^\pm(s)\dd Q_t^\pm(s) 
			+ \( \frac{ -\ii u^\pm(s) N_t^\pm(s) Q_{t-}^\pm(s)}{S_t} \) \dd S_t 
			+ \Big( \ii u^\pm(s) N_t^\pm(s) Q_{t-}^\pm(s) \Big) \dd B_t , \label{eq:dPi} \\
\Pi_0^\pm(s)
	&=  N_0^\pm(s) Q_0^\pm(s) . \label{eq:Pi0}
\end{align}
Then $\Pi^\pm(s)$ satisfies
\begin{align}
\Pi_T^\pm(s)
	&=	\ee^{ \ii s \<X\>_T } . \label{eq:Pi.T}
\end{align}
\end{proposition}

\begin{proof}
See \cite[Proposition 5.3]{rrvd}.
\end{proof}

We see from \eqref{eq:dPi} that the portfolio value $\Pi^\pm(s)$ is self-financing and involves trading in three assets: a European claim $Q^\pm(s)$, the underlying risky asset $S$, and a zero-coupon bond $B$.  Furthermore, from \eqref{eq:Pi.T}, we see that the portfolio value $\Pi^\pm$ replicates the exponential claim on realized quadratic variation $\ee^{ \ii s \<X\>_T }$.  Note that
\begin{align}
\Eb_t \ee^{\ii s \<X\>_T}
	&=	N_t^\pm(s) Q_t^\pm(s) 
	 = \Pi_t^\pm(s) , &
	&\text{(when $\rho = 0$).}  \label{eq:pi.note}
\end{align}
That is, the value of the replicating portfolio value equals the value of the claim.

Carr and Lee use Propositions \ref{thm:pricing} and \ref{thm:replication} to price and replicate more general claims with payoffs of the form $\varphi(\<X\>_T)$ where $\varphi$ can be expressed as a sum or integral of exponentials.  Examples of such functions include positive fractional powers $\varphi(\<X\>_T) = \<X\>_T^p$ ($0 < p < 1$), negative powers $\varphi(\<X\>_T) = \<X\>_T^{-r}$ ($r>0$), and puts $\varphi(\<X\>_T) = (\<X\>_T-K)^+$.  
For the purposes of this paper, we will consider only payoffs that can be expressed as a finite linear combination of exponentials
\begin{align}
\varphi(\<X\>_T)
	&=	\sum_k a_k \ee^{\ii s_k \<X\>_T } , &
a_k, s_k
	&\in \Cb , &
s_k
	&\neq \ii/8 .  \label{eq:linear-combo} 
\end{align}
The reason for the restriction $s_k \neq \ii/8$ will become clear below.
The following proposition states how a claim with a payoff of the form \eqref{eq:linear-combo}
can be priced and replicated.

\begin{proposition}[Pricing and Replication of general claims]
\label{thm:general}
Assume $\rho = 0$ and consider a claim with a payoff $\varphi(\<X\>_T)$ of the form \eqref{eq:linear-combo}.
The price at time $t$ of such a claim is
\begin{align}
\Eb_t \varphi(\<X\>_T)
	&=	\sum_k a_k N_t^\pm(s_k) Q_t^\pm(s_k) , 
\end{align}
where the processes $N^\pm(s)$ and $Q^\pm(s)$ are defined in \eqref{eq:NQ}.  Moreover, define the value of a self-financing portfolio as follows
\begin{align}
\dd \Pi_t^\pm
	&=	\sum_k a_k \dd \Pi_t^\pm(s_k), &
\Pi_0^\pm
	&=	\sum_k a_k \Pi_0^\pm(s_k) , \label{eq:Pi.basic}
\end{align}
where the differential $\dd \Pi_t^\pm(s)$ and the initial value $\Pi_0^\pm(s)$ are given by \eqref{eq:dPi} and \eqref{eq:Pi0}, respectively.  Then we have
\begin{align}
\Pi_T^\pm
	&=	\varphi(\<X\>_T) .
\end{align}
\end{proposition}

\begin{proof}
Use the fact that $\varphi$ is a linear combination of exponentials and apply Propositions \ref{thm:pricing} and \ref{thm:replication}.
\end{proof}

We shall refer to $\Pi^\pm$ as the \textit{basic} replicating portfolio value in order to distinguish it from the correlation immunized portfolio value that we discuss in the next section.


\subsection{Approximate pricing and replication under nonzero correlation}
\label{sec:corr}

In general, when $\rho \neq 0$, exact pricing and replication of volatility claims relative to European claims is not possible (the variance swap being a notable exception).  The goal of this section is to review the \textit{correlation immunization} strategy proposed in \cite{rrvd} in order to approximately price and replicate volatility claims when $\rho \neq 0$.  

A key observation of Carr and Lee is that, when $\rho = 0$, we have from \eqref{eq:pi.note} that
\begin{align}
\Eb_t \ee^{\ii s \<X\>_T}
	&=	\alpha^+ \Pi_t^+(s) + \alpha^- \Pi_t^-(s) , &
1
	&=	\alpha^+ + \alpha^- , &
	&\text{(when $\rho = 0$).} \label{eq:pi.combo}
\end{align}
Thus, when $\rho = 0$, not only does $\Pi^\pm(s)$ replicate the exponential claim, but so do linear combinations of the form \eqref{eq:pi.combo}.
Note, as $\rho$ moves away from zero, the value of the exponential claim $\Eb_t \ee^{\ii s \<X\>_T}$ does \textit{not} change because $\rho$ does not appear in the dynamics of $\<X\>$.
However, the value of the portfolio $\Pi_t^\pm(s) = N_t^\pm(s) Q_t^\pm(s)$ \textit{does} change because $\rho$ appears in the dynamics of $X$ and $Q_t^\pm(s) = \Eb_t \ee^{\ii u^\pm(s) X_T}$.
In general, there is an $\Oc(\rho)$ difference between the the true price of the claim and the value of the replicating portfolio 
\begin{align}
\Eb_t \ee^{\ii s \<X\>_T} - 
	&=	\alpha^+ \Pi_t^+(s) + \alpha^- \Pi_t^-(s) + \Oc(\rho), &
1
	&=	\alpha^+ + \alpha^- . \label{eq:asymptotic-2}
\end{align}
What Carr and Lee show, is that by defining $\alpha^\pm(s)$ as the unique solution to
\begin{align}
1
	&=	\alpha^+(s) + \alpha^-(s) , &
0
	&=	\alpha^+(s) u^+(s) + \alpha^-(s) u^-(s) , &
s
	&\neq \ii / 8 , \label{eq:alpha.s}
\end{align}
we have formally that
\begin{align}
\Eb_t \ee^{\ii s \<X\>_T}
	&=	\alpha^+(s) \Pi_t^+(s) + \alpha^-(s) \Pi_t^-(s) + \Oc(\rho^2) . \label{eq:asymptotic}
\end{align}
Note that we have excluded $s = \ii/8$ because, in this case, \eqref{eq:alpha.s} has no solution.
This is the reason for the restriction $s_k \neq \ii / 8$ in \eqref{eq:linear-combo}.
By choosing $\alpha^\pm = \alpha^\pm(s)$ we can reduce the pricing and replication error associated with strategies of the form $\alpha^+ \Pi^+(s) + \alpha^- \Pi^-(s)$ from $\Oc(\rho)$ to $\Oc(\rho^2)$.  This motivates the following definition.

\begin{definition}
\label{def:immune}
Consider a claim with a payoff $\varphi(\<X\>_T)$ of the form \eqref{eq:linear-combo}.
For such a claim, we define the \textit{correlation immunized hedging portfolio value}
$\Pi = (\Pi_t)_{0 \leq t \leq T}$ as follows
\begin{align}
\dd \Pi_t
	&=	\sum_k a_k \Big( \alpha^+(s_k) \dd \Pi_t^+(s_k) + \alpha^-(s_k) \dd \Pi_t^-(s_k) \Big), \label{eq:Pi.immunized} \\
\Pi_0
	&=	\sum_k a_k  \Big( \alpha^+(s_k) \Pi_0^+(s_k) + \alpha^-(s_k) \Pi_0^-(s_k) \Big) , 
\end{align}
where the differential $\dd \Pi_t^\pm(s)$ and the initial value $\Pi_0^\pm(s)$ are given by \eqref{eq:dPi} and \eqref{eq:Pi0}, respectively. 
\end{definition}

Denoting by $V_t := \Eb_t \varphi(\<X\>_T)$ the true price of the volatility derivative, it follow from \eqref{eq:asymptotic-2} that the basic hedging portfolio value $\Pi^\pm$ satisfies
\begin{align}
V_t 
	&=	\Pi_t^\pm + \Oc(\rho) , \label{eq:asymptotic-3}
\end{align}
whereas from \eqref{eq:asymptotic} the correlation immunized hedging portfolio value $\Pi$ satisfies
\begin{align}
V_t 
	&=	\Pi_t + \Oc(\rho^2) . \label{eq:asymptotic-4}
\end{align}
Equations \eqref{eq:asymptotic-3} and \eqref{eq:asymptotic-4} tell us how quickly the hedging portfolio values $\Pi^\pm$ and $\Pi$ deviate from the true value of the volatility derivative $V$ as $\rho$ moves away from zero.  What these equations do \textit{not} tell us, however, is how large the hedging errors can be for a \textit{fixed} $\rho$.  Obviously, from a practical point of view, one would like to know how large hedging errors can be for a fixed $\rho$.  We will investigate this issue in Section \ref{sec:monte-carlo} by performing a number of Monte Carlo simulations.  Before doing this, however, we discuss one other advantage of the correlation immunized pricing and hedging strategy that, to our knowledge, has not been discussed in literature.

\begin{proposition}
\label{thm:real}
Consider a claim with payoff \eqref{eq:linear-combo} where, for all $k$ we have $a_k \in \Rb$ and $s_k = - \ii \lam_k$ with $\lam_k \in \Rb$.  Then we have $\varphi(\<X\>_T) \in \Rb$ and, furthermore, the associated correlation immunized portfolio value $\Pi$, given by \eqref{eq:Pi.immunized}, satisfies $\Pi_t \in \Rb$ for all $t \in [0,T]$.
\end{proposition}

\begin{proof}
If $u^\pm$ are real, then so are $\alpha^\pm$.
If $u^\pm$ are not real, then they are complex conjugates.
But in that case, $N^\pm$, $Q^\pm$, and $\alpha^\pm $are all complex conjugates.
To see the latter fact, note that if ($\alpha^+$, $\alpha^-$) satisfy \eqref{eq:alpha.s},
then so do $\overline{\alpha^-}$, $\overline{\alpha^+}$.  Therefore,
uniqueness of the  \eqref{eq:alpha.s} solution implies $\alpha^-$ = $\overline{\alpha^+}$.
Hence, $\Pi$ is real.
\end{proof}

Note that, when $\rho \neq 0$, even when the derivative payoff is real $\varphi(\<X\>_T) \in \Rb$, the basic replicating portfolio value $\Pi^\pm$, defined by \eqref{eq:Pi.basic}, will in general be $\Cb$-valued.  Thus, in addition to reducing the pricing and hedging error from $\Oc(\rho)$ to $\Oc(\rho^2)$, another advantage of the correlation immunization strategy is that it ensures that the hedging portfolio value $\Pi$ is $\Rb$-valued when $\varphi(\<X\>_T) \in \Rb$.

%
%

\section{Monte Carlo experiments}
\label{sec:monte-carlo}
In order to test the correlation immunization pricing and hedging strategy described in Section \ref{sec:corr}, we assume that the volatility process $\sig$ has risk-neutral dynamics as described by \cite{heston1993}.  Specifically, the volatility process $\sig$ is modeled by a stochastic differential equation (SDE) of the form
\begin{align}
\sig_t
	&=	\sqrt{Y_t} , &
\dd Y_t
	&=	\kappa (\theta - Y_t)\dd t + \delta \sqrt{Y_t} \dd W^2_t , \label{eq:heston}
\end{align}
with $(X_0, Y_0) \in \Rb \times (0,\infty)$, $\kappa, \theta, \delta > 0$.
The advantage of assuming that $\sig$ has Heston dynamics is that the value of an European exponential claim $Q_t^\pm(s)$ defined in \eqref{eq:NQ} (i.e., the characteristic function of $X$) can be computed explicitly. 
We have
\begin{align}
Q_t^\pm(s) 
	:= \Eb_t \ee^{ \ii u^\pm(s) X_T}
	&= \ee^{\ii u^\pm(s) X_t + C(T-t, u^\pm(s)) + Y_t D(T-t, u^\pm(s))} , \label{eq:Q}
\end{align}
where the functions $C$ and $D$ are given by
\begin{align}
C(\tau, u)
    &:= \frac{\kappa \theta}{\delta^2}\((\kappa-\ii \rho \delta u + d(u))\tau - 2\log\[\frac{1-\gamma(u)\ee^{d(u) \tau}}{1-\gamma(u)}\]\) , \label{eq:C} \\
D(\tau, u)
    &:= \frac{\kappa - \ii \rho u \delta + d(u)}{\delta^2} \frac{1-\ee^{d(u) \tau}}{1- \gamma(u) \ee^{d(u) \tau}} , \label{eq:D} \\
\gamma(u)
    &:= \frac{\kappa - \ii \rho u \delta + d(u)}{\kappa - \ii \rho u \delta - d(u)} ,\\
d(u)
    &:= \sqrt{\delta^2 \(u^2 + \ii u\) + (\kappa - \ii \rho u \delta)^2} .
\end{align}
The characteristic function of $\<X\>_T$ can also be computed explicitly and is given by
\begin{align}
\Eb_t \ee^{\ii s \<X\>_T}
	&=	\ee^{\ii s \< X\>_t + A(T-t,s) + Y_t B(T-t,s)} , \label{eq:qv.char}
\end{align}
where the functions $A$ and $B$ are defined as follows
\begin{align}
A(\tau,s)
    &:= \frac{2\kappa \theta}{\delta^2}\log{\frac{2\xi \ee^{\frac{1}{2}\tau (\xi + \kappa)}}{\xi - \kappa + \ee^{\tau \xi (\xi+\kappa)}}} , &
B(\tau,s)
    &:= \frac{2\ii s (\ee^{\tau \xi} - 1)}{\xi - \kappa + \ee^{\tau \xi (\xi+\kappa)}}, &
\xi
    &:= \sqrt{\kappa^2 - 2\delta^2 \ii s}.
\end{align}
Using \eqref{eq:qv.char}, the value $V_t := \Eb_t \varphi( \<X\>_T)$ of a volatility derivative of the form \eqref{eq:linear-combo}
is given by
\begin{align}
V_t
	:=	\Eb_t \varphi( \<X\>_T )
	&=	\sum_k a_k \ee^{\ii s_k \< X\>_t + A(\tau,s_k) + Y_t B(\tau,s_k)} .
\end{align}
In all of the Monte Carlo simulations we perform, the following parameters remain fixed
\begin{align}
X_0
	&=	0 , &
Y_0
	&=	0.04 , &
\kappa
	&=	1.15 , &
\theta
	&=	0.04 , &
\delta
	&=	0.2 , &
T
	&=	1 .
\end{align}
We use a standard Euler-Maruyama discretization scheme with time step $\Delta t = 1/1,000$ and we generate {$N = 10,000$} sample paths.
Specifically, the $i^\text{th}$ sample path $(X^i,Y^i)$ is approximated using
\begin{align}
X_{t + \Delta t}^i 
	&=	X_t^i - \tfrac{1}{2} Y_t^i \Delta t + \sqrt{Y_t^i} \Big( \rhob ( W_{t + \Delta t}^{1,i} - W_t^{1,i} ) + \rho ( W_{t + \Delta t}^{2,i} - W_t^{2,i} ) \Big), \\
Y_{t + \Delta t}^i
	&=  Y_t^i + \kappa (\theta - Y_t^i) \Delta t + \del \sqrt{Y_t^i} ( W_{t + \Delta t}^{2,i} - W_t^{2,i} )  ,
\end{align}
where the increments $W_{t + \Delta t}^{j,i} - W_t^{j,i}$ are independent $\Nc(0,\Delta t)$ random variables.
Note that, while $Y^i$ can in theory become negative, we never encountered this in our simulations.
The $i^\text{th}$ sample path of quadratic variation $\<X^i\>$ is generated using
\begin{align}
\<X^i\>_{t+\Delta t}
	&=	\< X^i \>_t + ( X_{t + \Delta t}^i -  X_t^i )^2. \label{eq:qv.approx}
\end{align}
Finally, the $i^\text{th}$ sample path of the basic replicating portfolio value $\Pi^{\pm,i}$ is approximated using
\begin{align}
\Pi_{t+\Delta t}^{\pm,i} 
	&= \Pi_t^{\pm,i} + \sum_k a_k N_t^{\pm,i}(s_k) \Big( Q_{t+\Delta t}^{\pm,i}(s_k) - Q_t^{\pm,i}(s_k) \Big) \\ &\quad
			+ \sum_k a_k \( \frac{ -\ii u^\pm(s_k) N_t^{\pm,i}(s) Q_{t-}^{\pm,i}(s_k)}{S_t} \) \Big( S_{t+\Delta t}^i - S_t^i \Big) . 
\end{align}

We will be interested in comparing how closely the basic hedging portfolio values $\Pi^\pm$ and the correlation immunized portfolio value $\Pi$ replicate the derivative payoff $\varphi(\<X\>_T)$.  
To this end, we denote the hedging errors associated with the $i^\text{th}$ sample path by
\begin{align}
\eps^{\pm,i}
    &= {\Pi^{\pm,i}_T - \varphi(\<X^i\>_T)} , &
\eps^i
    &= { \Pi^i_T - \varphi(\<X^i\>_T) } . \label{eq:error-i}
\end{align}
From this, we compute the sample means and sample standard deviations of the hedging errors
\begin{align}
\widehat{\eps}^\pm
    &:= \frac{1}{N} \sum_{i=1}^N \eps^{\pm,i} , &
\widehat{\eps} 
    &:= \frac{1}{N} \sum_{i=1}^N \eps^i , \label{eq:mean} \\
\widehat{\sigma}^\pm
    &:= \Big( \frac{1}{N-1} \sum_{i=1}^N (\eps^{\pm,i}-\widehat{\eps}^\pm)^2 \Big)^{1/2} , &
\widehat{\sigma}
    &:= \Big( \frac{1}{N-1} \sum_{i=1}^N (\eps^i-\widehat{\eps})^2 \Big)^{1/2} . \label{eq:stand-dev} 
\end{align}
Below, we describe the results of our Monte Carlo experiments.


\subsection{Increasing exponential of realized variance}
\label{sec:exp}
In this section, we consider a volatility derivative with a simple increasing exponential payoff
\begin{align}
\varphi(\<X\>_T)
	&=	\ee^{\<X\>_T} . \label{eq:phi.exp}
\end{align}
Note that, for the payoff \eqref{eq:phi.exp}, we have $\Re \Pi_0^\pm = \Pi_0^\pm$ for all $\rho \in [-1,1]$.
In Figure \ref{fig:rho-effect-exp-pos} we plot $\Pi_0^\pm$, $\Pi_0$ and $V_0$ as functions of $\rho$.
Recall that $\Pi_0^\pm$ is the initial value of the hedging portfolio as computed \text{without} correlation immunization,
$\Pi_0$ is the initial value of the hedging portfolio as computed \text{with} correlation immunization,
and $V_0$ is the \textit{true} value of the volatility derivative.
Note that all four methods give the same price when $\rho = 0$ (as they ought to).
However, as $\rho$ moves away from zero, the four methods diverge, and the correlation immunization price $\Pi_0$ provides the best approximation of the true price $V_0$ for all values of $\rho \in [-1,1]$.

In Figure \ref{fig:sample-path-exp-pos} we plot sample paths of $\Pi^\pm$, $\Pi$ and $V$ for $\rho = \{-0.99, -0.66, 0, 0.99\}$.
From the figure, we see that when $\rho = 0$ all three replication strategies $\Pi^+$, $\Pi^-$ and $\Pi$ closely track the true value $V$ of the derivative.
Note that, because the results of Section \ref{sec:ind} are \textit{exact}, any tracking error in the $\rho = 0$ case is \textit{entirely} due to the discretization error associated with the Euler-Maruyama scheme.  When $\rho > 0$ we observe that $\Pi^+ < V$ and $\Pi^- > V$ while when $\rho < 0$ we observe that $\Pi^+ > V$ and $\Pi^- < V$.  

In Figure \ref{fig:hist-exp-pos} we plot histograms of the hedging errors $\eps^\pm$ and $\eps$ for $\rho = \{-0.99, -0.66, 0, 0.66, 0.99\}$.
It is clear from the histograms that, for all values of $\rho$, the hedging error of the correlation immunization strategy $\eps$ is centered near zero, whereas the hedging errors without immunization $\eps^\pm$ are typically centered away from zero.  
Summary statistics of the Monte Carlo simulations are provided in Table \ref{tbl:stats-exp-pos}.  
The table confirms what we observe visually from the figures; for each of the values of $\rho$ we tested, the hedging strategy that best approximates $V$ is the correlation immunization strategy $\Pi$.
Overall, compared to the basic strategies $\Pi^\pm$, it is clear that the correlation immunization strategy $\Pi$ dramatically reduces pricing and hedging errors for all values of $\rho$.


We have noted above that $\Pi^+ < V$ and $\Pi^- > V$ when $\rho > 0$ and $\Pi^+ > V$ and $\Pi^- < V$ when $\rho < 0$.
It is natural to wonder if these inequalities hold for all payoffs of the form $\varphi(\<X\>_T) = \ee^{c\<X\>_T}$ where $c>0$.
The following theorem addresses this question.

\begin{theorem}
\label{thm:d-rho}
Fix $c>0$.  Suppose that the dynamics of $\sig$ are given by \eqref{eq:heston} and
\begin{align}
-\del^2 \kappa^2 + 2 c \del^4 
	&> 0 .
\end{align}
Then we have
\begin{align}
\d_\rho Q_t^+(-\ii c) 
	&< 0 , &
	&\text{and}&
\d_\rho Q_t^-(-\ii c)
	&> 0 , &
	&\forall \, \rho < \overline{r} := \frac{ - \kappa }{ \del \ii u^-(-\ii c) } > 0 , \label{derivatives}
\end{align}
and thus $\Pi_t^\pm(-\ii c) := N_t^\pm(-\ii c) Q_t^\pm(-\ii c)$ and $V_t := \Eb_t \ee^{c \<X\>_T}$ satisfy
\begin{align}
\begin{aligned}
\Pi_t^+(-\ii c)
	&> V_t , &
\rho
	&< 0 , \\
\Pi_t^+(-\ii c)
	&< V_t , &
\overline{r} > \rho
	&> 0	, \\
 \Pi_t^-(-\ii c)
	&> V_t , &
\overline{r} > \rho
	&> 0 , \\
\Pi_t^-(-\ii c)
	&< V_t , &
\rho
	&< 0	,
\end{aligned} \label{inequalities}
\end{align}
\end{theorem}

\begin{proof}
The inequalities listed in \eqref{inequalities} follow from \eqref{derivatives} because
\begin{align}
\d_\rho N_t^\pm(s)
	&= 0 , &
\d_\rho V_t
	&= 0 , &
N_t^\pm(- \ii c)
	&> 0 ,
\end{align}
and when $\rho=0$ we have $\Pi_t^\pm(s) =	V_t$.
Therefore, we need only to check that \eqref{derivatives} holds.  Noting that $Q_t^\pm(s)$ is given explicitly by \eqref{eq:Q}, we have
\begin{align}
\d_\rho Q_t^\pm(s)
	&=	Q_t^\pm(s) \cdot \big( \d_\rho C(\tau,u^\pm(s)) + Y_t \d_\rho D(\tau,u^\pm(s)) \big) , \label{eq:dQ}
\end{align}
where the functions $C$ and $D$ are given by \eqref{eq:C} and \eqref{eq:D}, respectively.
Noting that $Q_t^\pm(-\ii c)>0$, we see from \eqref{eq:dQ} that the sign of $\d_\rho Q_t^\pm(-\ii c)$ is determined by the sign of $\d_\rho C(\tau,u^\pm(-\ii c)) + Y_t \d_\rho D(\tau,u^\pm(-\ii c))$.  It is straightforward to show that
\begin{align}
\d_\rho C(\tau,u)
	&=	\frac{-\ii u \del }{d^2(u)}\xi^C(\tau,u)  D(\tau,u) , &
\d_\rho D(\tau,u)
	&=	\xi^D(\tau,u) \frac{ -\ii u \del^3 D^2(\tau,u)}{\alpha(u) d(u) (\ee^{d(u)\tau}-1)} , 
\end{align}
where we have defined
\begin{align}
\xi^C(\tau,u)
	&:=	\kappa \theta \( 2 - \frac{d(u) \tau (\ee^{d(u) \tau}+1)}{\ee^{d(u) \tau} - 1} \) , \\
\xi^D(\tau,u)
	&:= ( \beta(u) + d(u) ) (\ee^{2 d(u)\tau} - 1 ) - 2 d(u) ( \ee^{d(u)\tau} - 1 ) - 2 \beta(u) d(u) \tau \ee^{d(u) \tau} , \\
\alpha(u)
	&:=	\del^2 ( u^2 + \ii u ) , \\
\beta(u)
	&:= \kappa - \ii \rho u \del , \\
d(u)
	&:=	\sqrt{ \alpha(u) + \beta^2(u) } .
\end{align}
Using the above expressions, one can verify that
\begin{align}
\d_\rho C(\tau,u^+(-\ii c))
	&<0 , &
\d_\rho D(\tau,u^+(-\ii c))
	&<0 , &
	&\forall \, \rho < \overline{r} , \\
\d_\rho C(\tau,u^-(-\ii c))
	&>0 , &
\d_\rho D(\tau,u^-(-\ii c))
	&>0 , &
	&\forall \, \rho < \overline{r} .
\end{align}
Thus, noting that $Y_t > 0$ we have that \eqref{derivatives} holds.
\end{proof}


\subsection{Decreasing exponential of realized variance}
\label{sec:exp.2}
In this section, we consider a volatility derivative with a simple decreasing exponential payoff
\begin{align}
\varphi(\<X\>_T)
	&=	\ee^{-\<X\>_T} . \label{eq:phi.exp.2}
\end{align}
Note that, for the payoff \eqref{eq:phi.exp.2}, we have $\Pi_0^\pm \notin \Rb$ when $\rho \neq 0$.
In Figure \ref{fig:rho-effect-exp-neg} we plot the real and imaginary parts of $\Pi_0^\pm$, $\Pi_0$ and $V_0$ as functions of $\rho$.
We observe from the figure that $\Re \Pi_0^+ = \Re \Pi_0^-$ and that $\Im \Pi_0^+ = - \Im \Pi_0^-$.  In fact, this can be proven mathematically.
It is the imaginary part of $\Pi_0^\pm$, together with the imaginary part of $\alpha^\pm$, that contribute to $\Pi_0$ being better approximation of $V_0$ than $\Re \Pi_0^\pm$.

In Figure \ref{fig:sample-path-exp-neg} we plot sample paths of $\Re \Pi^\pm$, $\Pi$ and $V$ for $\rho = \{-0.99, -0.66, 0, 0.99\}$.
From the figure, we see that when $\rho = 0$ all three replication strategies $\Re \Pi^+$, $\Pi^-$ and $\Pi$ closely track the true value $V$ of the derivative.
Note once again that, because the results of Section \ref{sec:ind} are \textit{exact}, any tracking error in the $\rho = 0$ case is \textit{entirely} due to the discretization error associated with the Euler-Maruyama scheme.  
When $\rho > 0$ we observe that $\Re \Pi^\pm > V$ while when $\rho < 0$ we observe that $\Re \Pi^\pm < V$.  

In Figure \ref{fig:hist-exp-neg} we plot histograms of the hedging errors $\Re \eps^\pm$ and $\eps$ for $\rho = \{-0.99, -0.66, 0, 0.66, 0.99\}$.
It is clear from the histograms that, for all values of $\rho$, the hedging error of the correlation immunization strategy $\eps$ is centered near zero (though the mean is always positive when $\rho \neq 0$), whereas the hedging errors without immunization $\eps^\pm$ are typically centered away from zero (with means that are opposite the sign of $\rho$.).  
Summary statistics of the Monte Carlo simulations are provided in Table \ref{tbl:stats-exp-neg}.  
The table confirms what we observe visually from the figures; for each of the values of $\rho$ we tested, the hedging strategy that best approximates $V$ is the correlation immunization strategy $\Pi$.
Overall, compared to the basic strategies $\Pi^\pm$, it is clear that the correlation immunization strategy $\Pi$ reduces pricing and hedging errors for all values of $\rho$.


\subsection{Approximate put option on realized variance}
\label{sec:put}

A put option on realized variance with strike $K$ has a payoff $( K - \<X\>_T )^+$.  Carr and Lee provide a method of approximating put payoffs uniformly using Bernstein polynomials with exponential arguments.   We repeat Proposition 7.12 of \cite{rrvd} here, slightly modified for our purposes.

\begin{proposition}
\label{thm:bernstein}
Consider a payoff function $h \in C[0,\infty]$ such that $\lim_{v \to \infty} h(v)$ exists.
Define $h^* : [0,1] \to \Rb$ by $h^*(0) := h(\infty)$ and $h^*(x) := h( - (1/c) \log x )$ for $x>0$.  
Define $B_n$, the $n^\textup{th}$ Bernstein approximation of $h^*$ by
\begin{align}
B_n(x)
	&:=	\sum_{k=0}^n b_{n,k} x^k , &
b_{n,k}
	:=	\sum_{j=0}^k h^*(j/n) \binom{n}{k} \binom{k}{j}(-1)^{j-k} .
\end{align}
Then we have
\begin{align}
h(v)
	&=	\lim_{n \to \infty} B_n(\ee^{-c v}) ,
\end{align}
uniformly in $v \in [0,\infty)$.
\end{proposition}

\begin{proof}
See the proof of Proposition 7.12 of \cite{rrvd}.
\end{proof}

In this section, we consider an \textit{approximate} put payoff
\begin{align}
\varphi( \<X\>_T )
	&=	B_n(\ee^{- c \<X\>_t} ) , &
h(v)
	&=	( K - v )^+ , &
K
	&=	0.04 , &
c
	&=	10 , &
n
	&=	20 , \label{eq:put.payoff}
\end{align}
where $B_n$ is defined from $h$ as described in Proposition \ref{thm:bernstein}.
Note that $\varphi(\<X\>_T)$ is a payoff of the form \eqref{eq:linear-combo}.  
In Figure \ref{fig:approx-put}, in order to see how well the put approximation performs in the region where $\<X\>_T$ is likely to be, we plot $\varphi( \<X\>_T )$, given by \eqref{eq:put.payoff}, and $(K - \<X\>_T )^+$ as functions of $\<X\>_T$ as well as the density of $\<X\>_T$, which can be computed numerically from \eqref{eq:qv.char}.
In general, we see good agreement between $\varphi(\<X\>_T)$ and $(K - \<X\>_T)^+$, with the approximation deteriorating slightly near $\<X\>_T = K$.

Note that, for the payoff \eqref{eq:put.payoff}, we have $\Pi_0^\pm \notin \Rb$ when $\rho \neq 0$.
In Figure \ref{fig:rho-effect-put} we plot the real and imaginary parts of $\Pi_0^\pm$, $\Pi_0$ and $V_0$ as functions of $\rho$.
Once again, we observe from the figure that $\Re \Pi_0^+ = \Re \Pi_0^-$ and that $\Im \Pi_0^+ = - \Im \Pi_0^-$; 
this is due to the fact that $\varphi(\<X\>_T)$ is a linear combination of decreasing exponentials.
Interestingly, while $\Pi_0$ provides a better approximation of $V_0$ than $\Re \Pi_0^\pm$ when $\rho < 0$, we see that $\Re \Pi_0^\pm$ provides a better approximation of $V_0$ than $\Pi_0$ for some values of $\rho > 0$.

In Figure \ref{fig:sample-path-put} we plot sample paths of $\Re \Pi^\pm$, $\Pi$ and $V$ for $\rho = \{-0.99, -0.66, 0, 0.99\}$.
Consistent with Figure \ref{fig:rho-effect-put}, we see that both $\Pi$ and $\Re \Pi^\pm$ over-replicate $V$.
When $\rho < 0$, the correlation immune portfolio value $\Pi$ more closely tracks the true derivative value $V$ than do the real parts of the basic replication strategies $\Re \Pi^\pm$.  However, the improvement is not as drastic as for the increasing exponential payoff described in Section \ref{sec:exp}.
When $\rho > 0$, the real part of the basic replication strategy values $\Re \Pi^\pm$ more closely track the true derivative value $V$ than does the correlation immune portfolio value $\Pi$.

In Figure \ref{fig:hist-put} we plot histograms of the hedging errors $\Re \eps^\pm$ and $\eps$ for $\rho = \{-0.99, -0.66, 0, 0.99\}$.
When $\rho \neq 0$ we observe that all three replicating strategies ($\Pi$ and $\Re \Pi^\pm$) tend to over-replicate the option payoff.
When $\rho \leq 0$, the correlation immunized strategy $\Pi$ has a smaller hedging error (on average) than the basic replication strategies $\Re \Pi^\pm$.  
However, when $\rho \geq 0$, the basic replication strategies $\Re \Pi^\pm$ have a smaller hedging error (on average) than the correlation immunized strategy $\Pi$.
Summary statistics of the Monte Carlo simulations are provided in Table \ref{tbl:stats-put}.
Overall, use of the correlation immunization strategy $\Pi$ is recommended when $\rho \leq 0$.


\subsection{Approximate floating leg of a volatility swap}
\label{sec:vol-swap}

Carr and Lee note in the proof of Proposition 6.6 of \cite{rrvd} that the floating leg of a volatility swap can be written as an integral transform
\begin{align}
\varphi(\<X\>_T)
	&:=	\sqrt{\<X\>_T}
	 =	\frac{1}{\sqrt{2 \pi}} \int_0^\infty \frac{1 - \ee^{-z \<X\>_T}}{z^{3/2}} \dd z , \label{eq:vol.swap}
\end{align}
and use this to develop pricing and replication strategies (both robust and basic) for a derivative that pays the square root of realized volatility.
In this section, we consider the following \textit{approximate} square root payoff
\begin{align}
\varphi( \<X\>_T )
	&=	B_n(\ee^{- c \<X\>_t} ) , &
h(v)
	&=	\sqrt{v} , &
c
	&=	{10} , &
n
	&=	{20} , \label{eq:swap.payoff}
\end{align}
where $B_n$ is defined from $h$ as described in Proposition \ref{thm:bernstein}.
Note that the payoff in \eqref{eq:swap.payoff}, unlike the payoff in \eqref{eq:vol.swap}, is of the form \eqref{eq:linear-combo}.
In Figure \ref{fig:approx-swap}, in order to see how well the square root approximation performs, 
we plot $\varphi( \<X\>_T )$, given by \eqref{eq:swap.payoff}, and $\sqrt{\<X\>_T}$ as functions of $\<X\>_T$ as well as a density of $\<X\>_T$.
In general, we see good agreement between $\varphi(\<X\>_T)$ and $\<X\>_T$, especially in the regions where the density of $\<X\>_T$ is largest.

Note that, for the payoff \eqref{eq:put.payoff}, we have $\Pi_0^\pm \notin \Rb$ when $\rho \neq 0$.
In Figure \ref{fig:rho-effect-swap} we plot the real and imaginary parts of $\Pi_0^\pm$, $\Pi_0$ and $V_0$ as functions of $\rho$.
Once again, we observe from the figure that $\Re \Pi_0^+ = \Re \Pi_0^-$ and that $\Im \Pi_0^+ = - \Im \Pi_0^-$; 
this is due to the fact that $\varphi(\<X\>_T)$ is a linear combination of decreasing exponentials.
We note that when $\rho < 0$ we have $\Re \Pi_0^\pm < V_0$ and when $\rho > 0$ we have $\Re \Pi_0^\pm > V_0$.
We also see that $\Pi_0$ provides a better approximation of $V_0$ than $\Re \Pi_0^\pm$ for all values of $\rho$.

In Figure \ref{fig:sample-path-swap} we plot sample paths of $\Re \Pi^\pm$, $\Pi$ and $V$ for $\rho = \{-0.99, -0.66, 0, 0.99\}$.
Consistent with Figure \ref{fig:rho-effect-put}, we see that both $\Pi$ over-replicate $V$ when $\rho \neq 0$.
We also see that $\Re \Pi^\pm$ over-replicates $V$ when $\rho > 0$ and under-replicates $V$ when $\rho$ is negative.
For the particular sample path plotted, the correlation immunized portfolio value $\Pi$ more closely tracks $V$ than do the real parts of the basic portfolio values $\Re \Pi^\pm$ except when $\rho = -0.99$.  However, for other sample paths (not pictured), we have observed that $\Pi$ more closely tracks $V$ than do the real parts of the basic portfolio values $\Re \Pi^\pm$ for all values of $\rho$.

In Figure \ref{fig:hist-swap} we plot histograms of the hedging errors $\Re \eps^\pm$ and $\eps$ for $\rho = \{-0.99, -0.66, 0, 0.99\}$.
Consistent with Figure \ref{fig:rho-effect-swap} we see that the basic replicating portfolio values $\Re \Pi_T^\pm$ tend to under-replicate (over-replicate) the payoff $\varphi(\<X\>_T)$ when $\rho < 0$ ($\rho > 0$).  When $\rho \neq 0$ the correlation immunized portfolio value $\Pi_T$ slightly over-replicates the payoff $\varphi(\<X\>_T)$.
Though, for all values of $\rho$, the average hedging error of the correlation immunized portfolio value $\Pi_T$ is smaller than the hedging errors of the real parts of the basic heading portfolio values $\Re \Pi_T^\pm$.    Summary statistics of the Monte Carlo simulations are provided in Table \ref{tbl:stats-swap}.
Overall, compared to the value of the basic strategies $\Pi^\pm$, it is clear that the value of the correlation immunization strategy $\Pi$ dramatically reduces pricing and hedging errors for all values of $\rho$.

%
%

\section{Conclusion}
\label{sec:conclusion}
In this paper we have presented the results of a numerical investigation of Carr and Lee's correlation immunization strategy for volatility derivatives.  The results of our investigation confirm that the correlation immunization strategy is an effective way to minimize pricing and hedging errors that result from nonzero correlation between the underlying asset and the volatility process.  Additionally, we have proved that value of the correlation immunized portfolio is real-valued when the derivative payoff is real-valued.  This result provides further motivation to use the correlation immunized pricing and replication strategies rather than the basic strategies.

\subsection*{Acknowledgments}
The authors would like to express their gratitude to the anonymous referees and editor, whose feedback greatly improved this paper.  In particular, one referee provided the very concise proof of Proposition \ref{thm:real}, which appears in this version of the paper.

%
%

\bibliographystyle{chicago}
\bibliography{Bibtex-Master-3.05}

\begin{figure}[ht]
\centering
\begin{tabular}{c}
\includegraphics[width=0.45\textwidth]{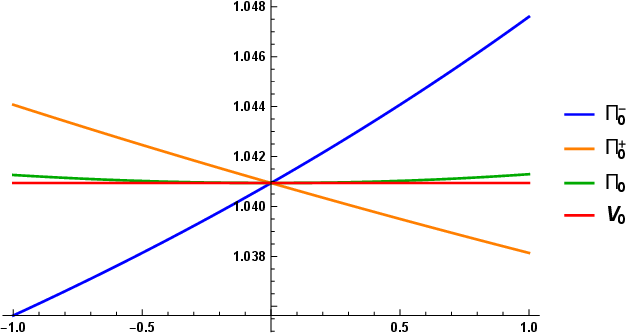} 
\end{tabular}
\caption{
A plot of $\Pi_0^\pm$, $\Pi_0$ and $V_0$ as functions of $\rho$ for the volatility derivative with payoff \eqref{eq:phi.exp}.
}
\label{fig:rho-effect-exp-pos}
\end{figure}

\begin{figure}[ht]
\centering
\begin{tabular}{cc}
\includegraphics[width=0.45\textwidth]{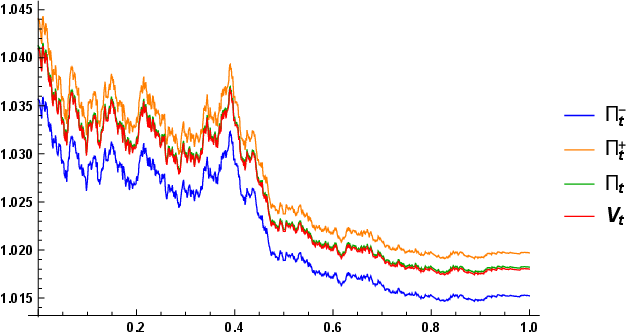} & 
\includegraphics[width=0.45\textwidth]{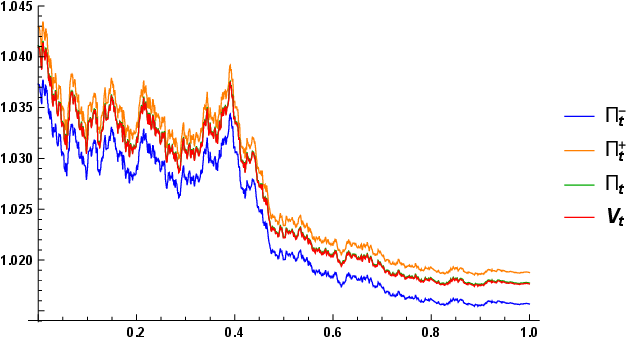}  \\
$\rho=-0.99$ & $\rho=-0.66$ \\
\includegraphics[width=0.45\textwidth]{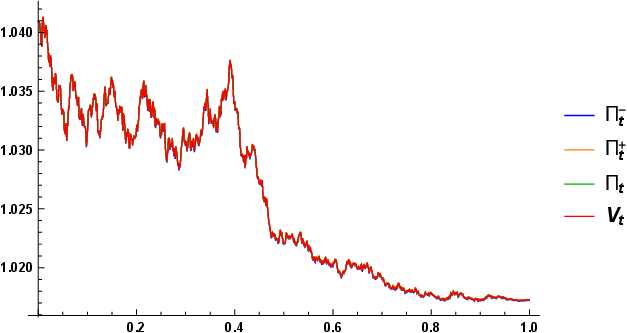} &
\includegraphics[width=0.45\textwidth]{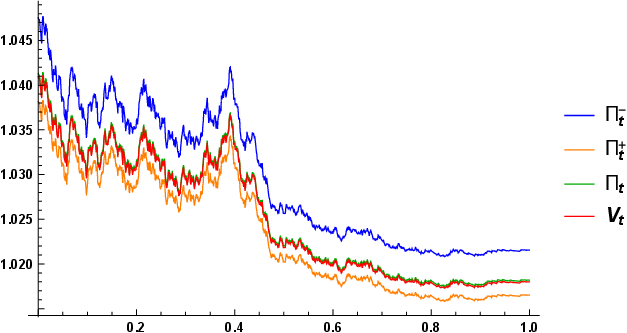} \\
$\rho=0$ & $\rho=0.99$
\end{tabular}
\caption{
Sample paths of $\Pi^\pm$, $\Pi$ and $V$ for the volatility derivative with payoff \eqref{eq:phi.exp}.
}
\label{fig:sample-path-exp-pos}
\end{figure}

\begin{figure}[ht]
\centering
\begin{tabular}{cc}
\includegraphics[width=0.45\textwidth]{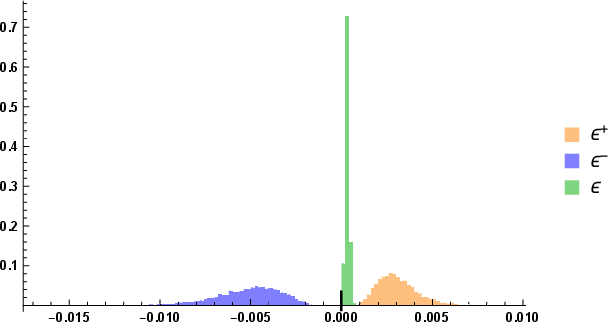} & 
\includegraphics[width=0.45\textwidth]{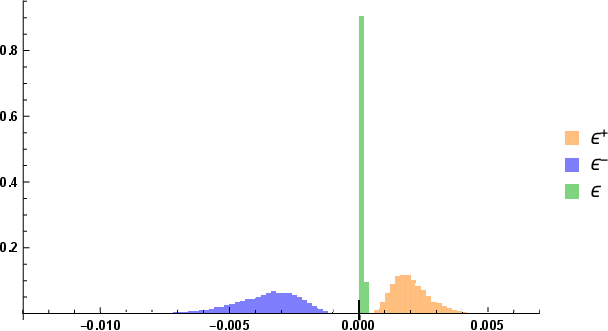}  \\
$\rho=-0.99$ & $\rho=-0.66$ \\
\includegraphics[width=0.45\textwidth]{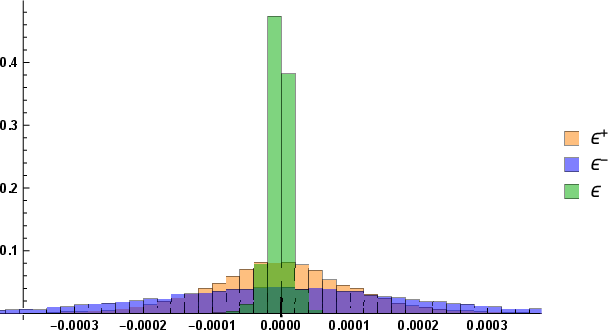} &
\includegraphics[width=0.45\textwidth]{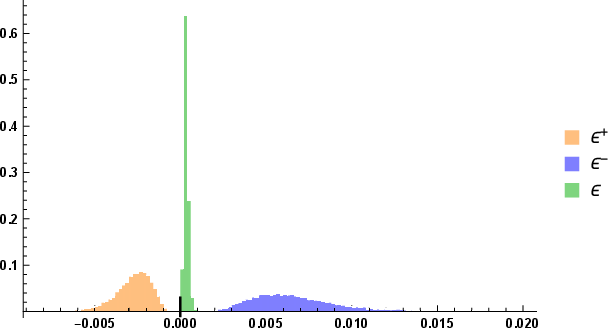} \\
$\rho=0$ & $\rho=0.99$
\end{tabular}
\caption{
Probability histogram of hedging errors $\eps^\pm$ and $\eps$ for the volatility derivative with payoff \eqref{eq:phi.exp}.
}
\label{fig:hist-exp-pos}
\end{figure}


\begin{figure}[ht]
\centering
\begin{tabular}{cc}
\includegraphics[width=0.45\textwidth]{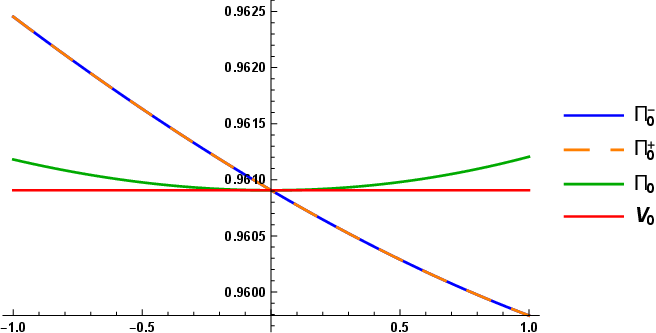} &
\includegraphics[width=0.45\textwidth]{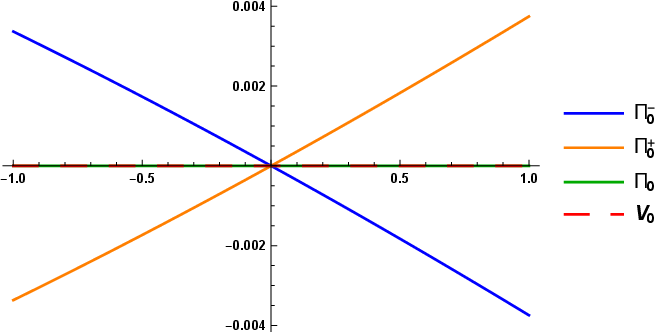} \\
Real part & Imaginary part
\end{tabular}
\caption{
A plot of the real and imaginary parts of $\Pi_0^\pm$, $\Pi_0$ and $V_0$ as functions of $\rho$ for the volatility derivative with payoff \eqref{eq:phi.exp.2}.
}
\label{fig:rho-effect-exp-neg}
\end{figure}

\begin{figure}[ht]
\centering
\begin{tabular}{cc}
\includegraphics[width=0.45\textwidth]{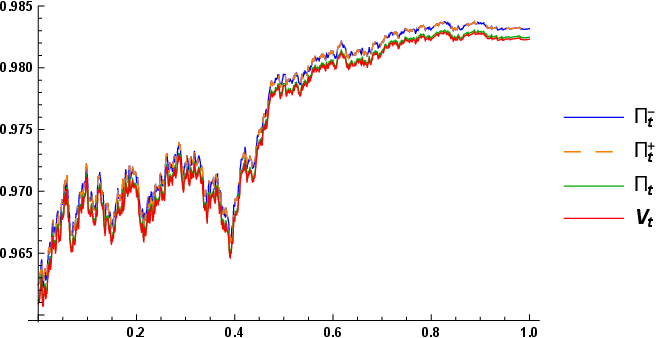} & 
\includegraphics[width=0.45\textwidth]{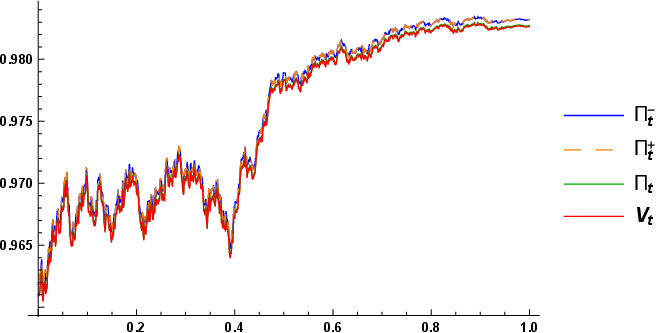}  \\
$\rho=-0.99$ & $\rho=-0.66$ \\
\includegraphics[width=0.45\textwidth]{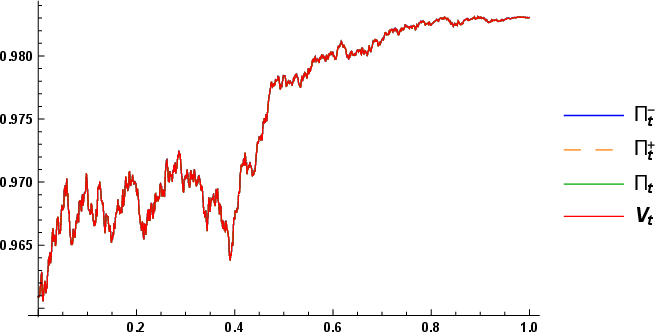} &
\includegraphics[width=0.45\textwidth]{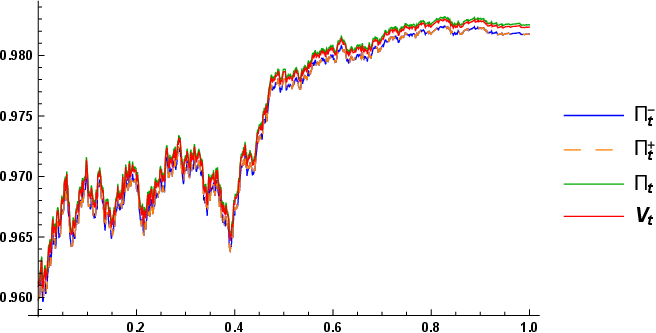} \\
$\rho=0$ & $\rho=0.99$
\end{tabular}
\caption{
Sample paths of $\Re \Pi^\pm$, $\Pi$ and $V$ for the volatility derivative with payoff \eqref{eq:phi.exp.2}.
}
\label{fig:sample-path-exp-neg}
\end{figure}

\begin{figure}[ht]
\centering
\begin{tabular}{cc}
\includegraphics[width=0.45\textwidth]{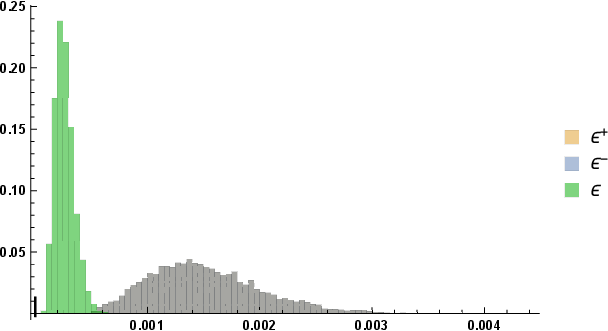} & 
\includegraphics[width=0.45\textwidth]{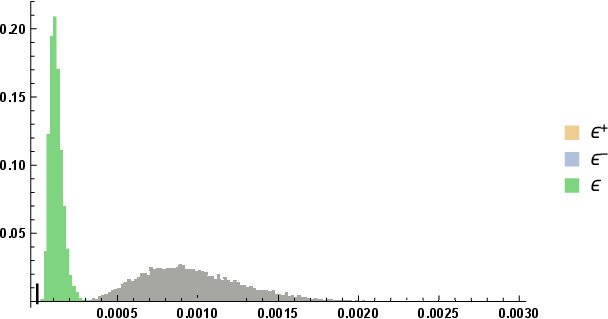}  \\
$\rho=-0.99$ & $\rho=-0.66$ \\
\includegraphics[width=0.45\textwidth]{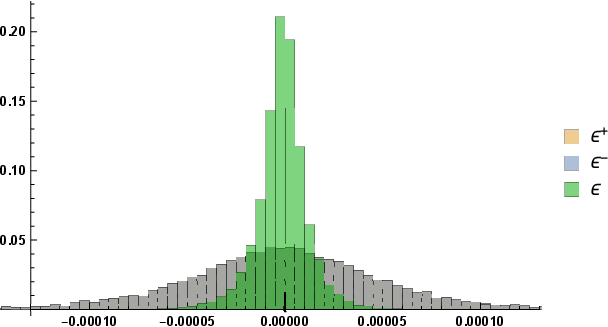} &
\includegraphics[width=0.45\textwidth]{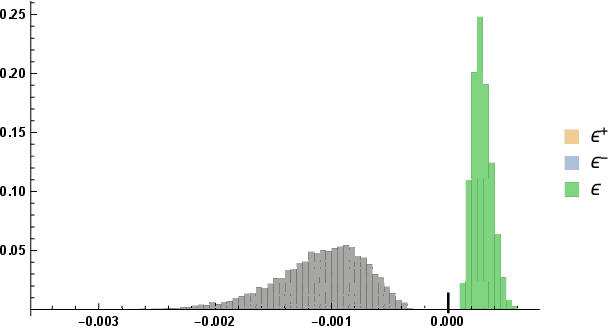} \\
$\rho=0$ & $\rho=0.99$
\end{tabular}
\caption{
Probability histogram of hedging errors $\Re \eps^\pm$ and $\eps$ for the volatility derivative with payoff \eqref{eq:phi.exp.2}.
Note that the gray histogram results from $\Re \eps^+ = \Re \eps^-$.
}
\label{fig:hist-exp-neg}
\end{figure}


\begin{figure}[ht]
\centering
\begin{tabular}{c}
\includegraphics[width=0.6 \textwidth]{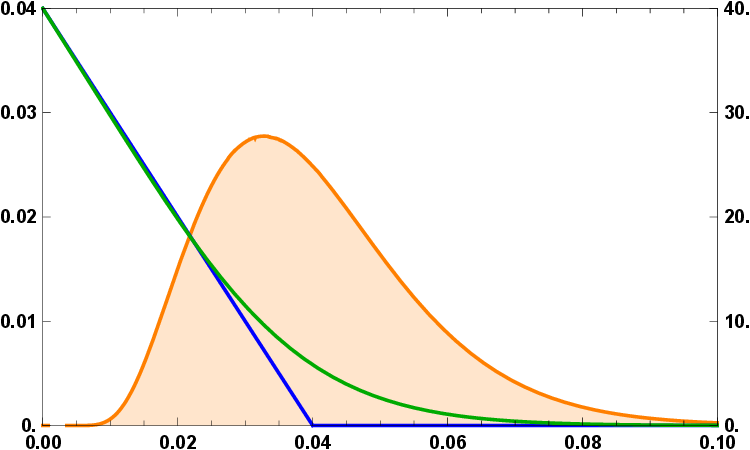}
\end{tabular}
\caption{
A plot of a put payoff $(K - \<X\>_T)^+$ (blue) and its approximation $\varphi(\<X\>_T)$, given by \eqref{eq:put.payoff} (green).
In the background, we plot the probability density of $\<X\>_T$.
The scale of the payoff functions is given on the left vertical axis and the scale of the density is given on the right vertical axis.
}
\label{fig:approx-put}
\end{figure}

\begin{figure}[ht]
\centering
\begin{tabular}{cc}
\includegraphics[width=0.45\textwidth]{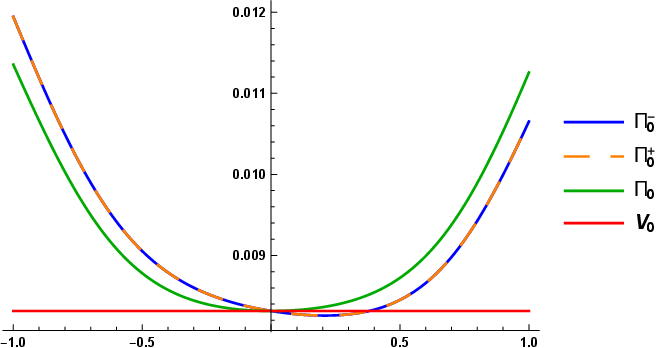} &
\includegraphics[width=0.45\textwidth]{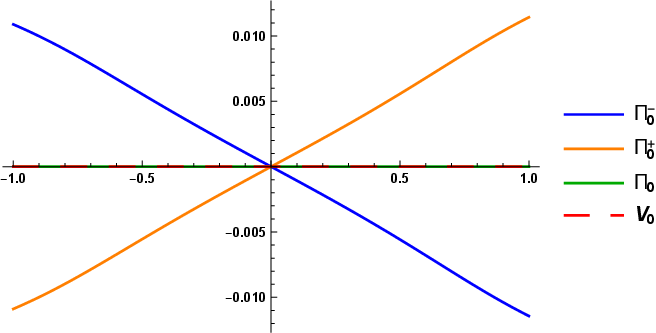} \\
Real part & Imaginary part
\end{tabular}
\caption{
A plot of the real and imaginary parts of $\Pi_0^\pm$, $\Pi_0$ and $V_0$ as functions of $\rho$ for the volatility derivative with payoff \eqref{eq:put.payoff}.
}
\label{fig:rho-effect-put}
\end{figure}

\begin{figure}[ht]
\centering
\begin{tabular}{cc}
\includegraphics[width=0.45\textwidth]{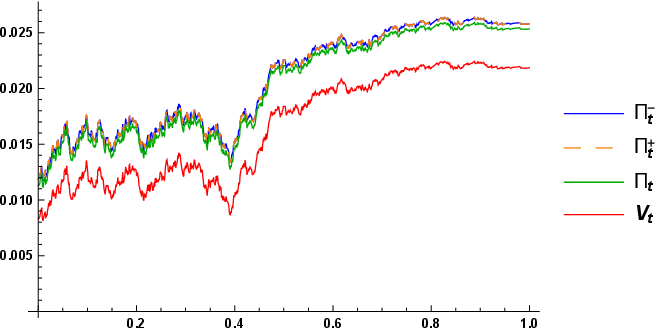} & 
\includegraphics[width=0.45\textwidth]{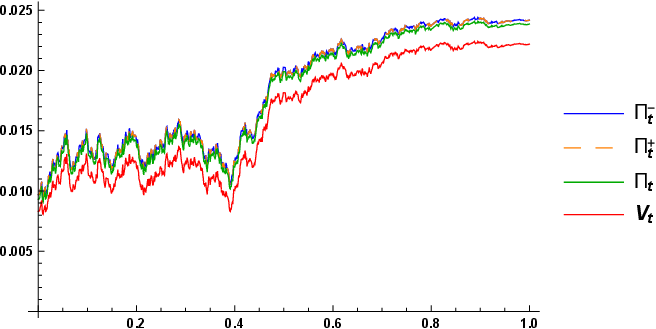}  \\
$\rho=-0.99$ & $\rho=-0.66$ \\
\includegraphics[width=0.45\textwidth]{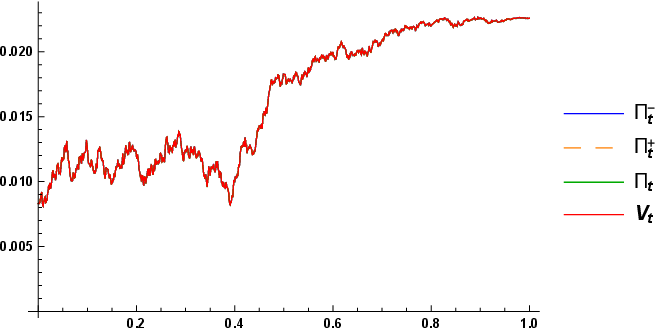} &
\includegraphics[width=0.45\textwidth]{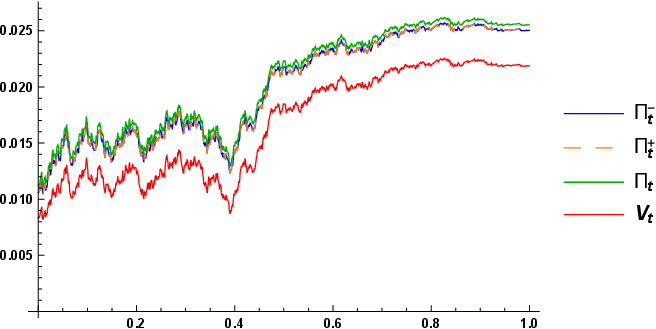} \\
$\rho=0$ & $\rho=0.99$
\end{tabular}
\caption{
Sample paths of $\Re \Pi^\pm$, $\Pi$ and $V$ for the volatility derivative with payoff \eqref{eq:put.payoff}.
}
\label{fig:sample-path-put}
\end{figure}

\begin{figure}[ht]
\centering
\begin{tabular}{cc}
\includegraphics[width=0.45\textwidth]{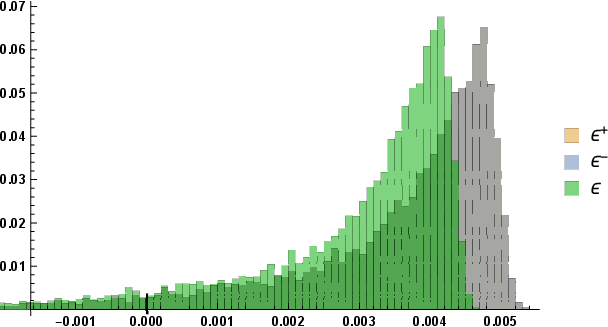} & 
\includegraphics[width=0.45\textwidth]{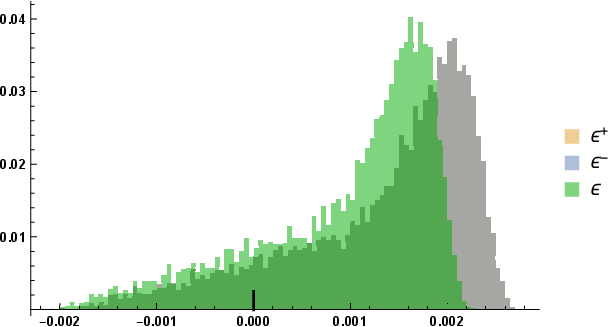}  \\
$\rho=-0.99$ & $\rho=-0.66$ \\
\includegraphics[width=0.45\textwidth]{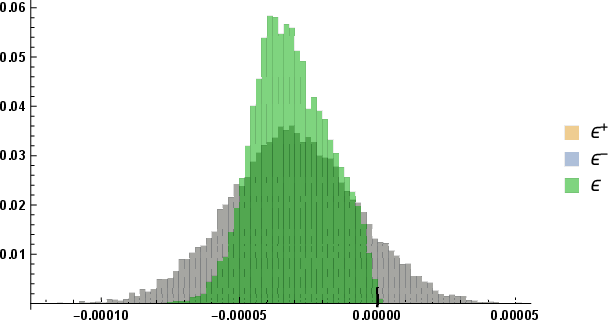} &
\includegraphics[width=0.45\textwidth]{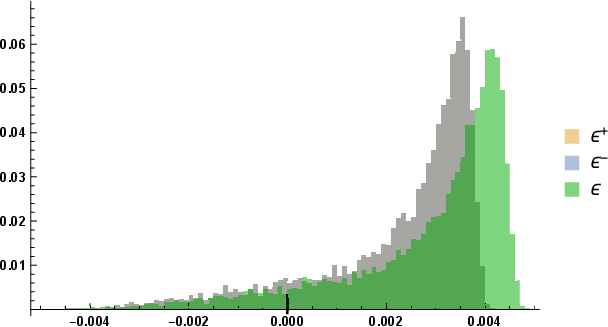} \\
$\rho=0$ & $\rho=0.99$
\end{tabular}
\caption{
Probability histogram of hedging errors $\Re \eps^\pm$ and $\eps$ for the volatility derivative with payoff \eqref{eq:put.payoff}.
Note that the gray histogram results from $\Re \eps^+ = \Re \eps^-$.
}
\label{fig:hist-put}
\end{figure}


\begin{figure}[ht]
\centering
\begin{tabular}{c}
\includegraphics[width=0.6 \textwidth]{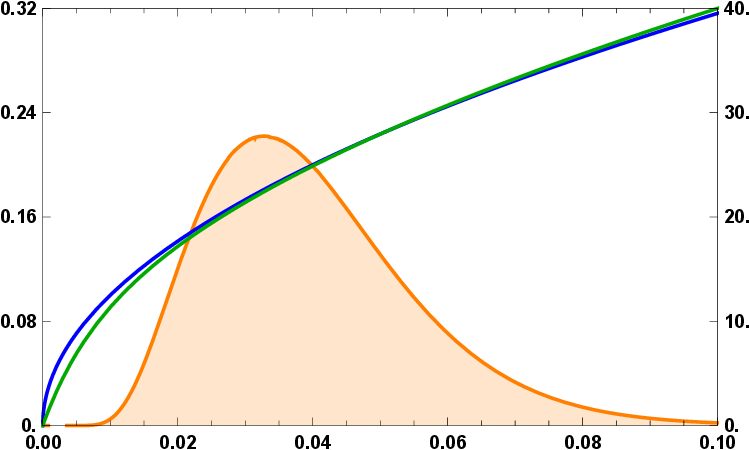}
\end{tabular}
\caption{
A plot of a square root payoff $\sqrt{\<X\>_T}$ (blue) and its approximation $\varphi(\<X\>_T)$, given by \eqref{eq:swap.payoff} (green).
In the background, we plot the probability density function of $\<X\>_T$.
The vertical axis on the left gives the scale of the payoffs and the vertical axis on the right gives the scale of the density.
}
\label{fig:approx-swap}
\end{figure}

\begin{figure}[ht]
\centering
\begin{tabular}{cc}
\includegraphics[width=0.45\textwidth]{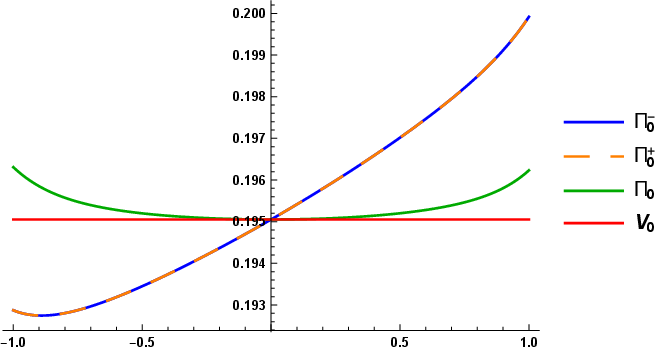} &
\includegraphics[width=0.45\textwidth]{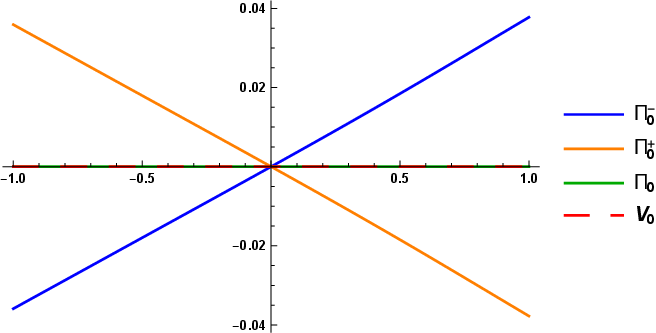} \\
Real part & Imaginary part
\end{tabular}
\caption{
A plot of the real and imaginary parts of $\Pi_0^\pm$, $\Pi_0$ and $V_0$ as functions of $\rho$ for the volatility derivative with payoff \eqref{eq:swap.payoff}.
}
\label{fig:rho-effect-swap}
\end{figure}

\begin{figure}[ht]
\centering
\begin{tabular}{cc}
\includegraphics[width=0.45\textwidth]{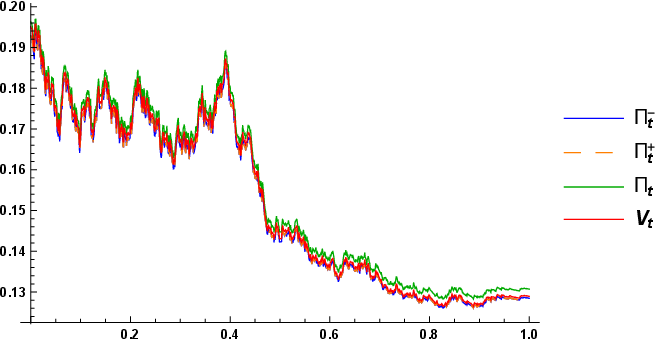} & 
\includegraphics[width=0.45\textwidth]{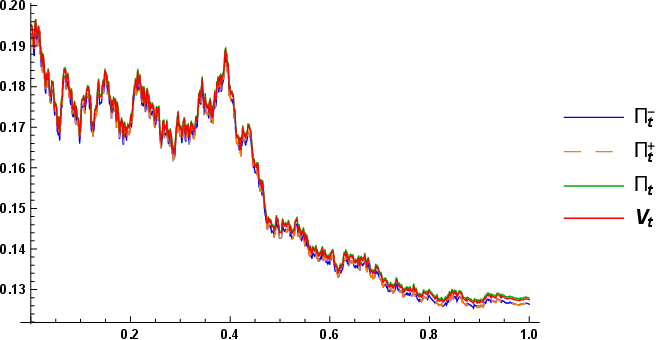}  \\
$\rho=-0.99$ & $\rho=-0.66$ \\
\includegraphics[width=0.45\textwidth]{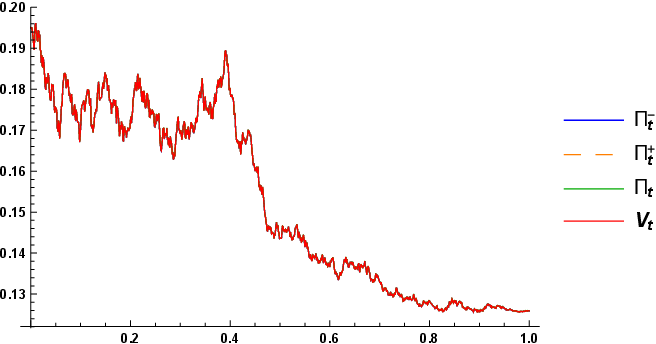} &
\includegraphics[width=0.45\textwidth]{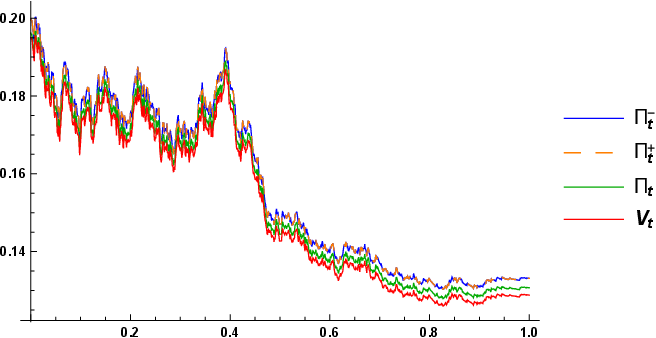} \\
$\rho=0$ & $\rho=0.99$
\end{tabular}
\caption{
Sample paths of $\Re \Pi^\pm$, $\Pi$ and $V$ for the volatility derivative with payoff \eqref{eq:swap.payoff}.
}
\label{fig:sample-path-swap}
\end{figure}

\begin{figure}[ht]
\centering
\begin{tabular}{cc}
\includegraphics[width=0.45\textwidth]{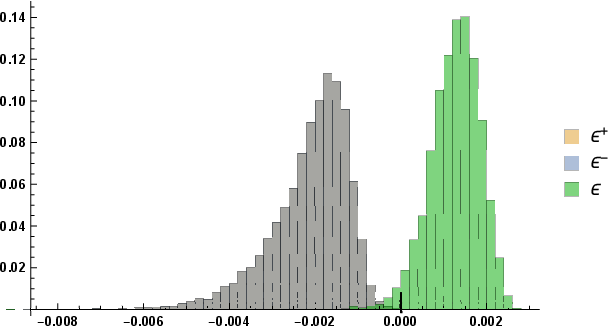} & 
\includegraphics[width=0.45\textwidth]{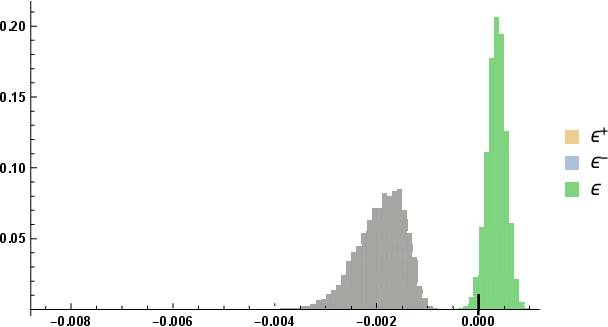}  \\
$\rho=-0.99$ & $\rho=-0.66$ \\
\includegraphics[width=0.45\textwidth]{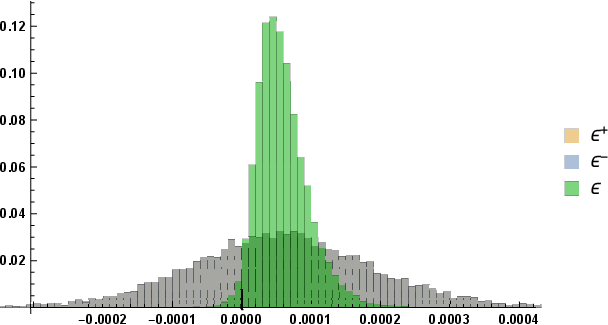} &
\includegraphics[width=0.45\textwidth]{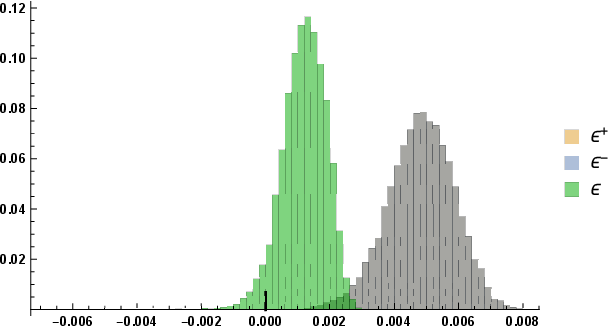} \\
$\rho=0$ & $\rho=0.99$
\end{tabular}
\caption{
Probability histogram of hedging errors $\Re \eps^\pm$ and $\eps$ for the volatility derivative with payoff \eqref{eq:swap.payoff}.
Note that the gray histogram results from $\Re \eps^+ = \Re \eps^-$.
}
\label{fig:hist-swap}
\end{figure}

%
%

\clearpage


\begin{table}[ht]
\centering
\begin{tabular}{c|ccccc}
\hline
                       & $\rho=-0.99$ & $\rho=-0.66$ & $\rho=0$  & $\rho=0.66$ & $\rho=0.99$ \\ \hline
$\widehat{\eps}^-$ 		 & -5.23E-03    & -3.64E-03    & -4.08E-06 & 4.24E-03    & 6.57E-03    \\
$\widehat{\eps}$   		 & 3.10E-04     & 1.39E-04     & -2.80E-06 & 1.49E-04    & 3.43E-04    \\
$\widehat{\eps}^+$ 		 & 3.08E-03     & 2.03E-03     & -2.16E-06 & -1.90E-03   & -2.77E-03   \\ \hline
$\widehat{\sigma}^-$   & 1.91E-03     & 1.35E-03     & 2.29E-04  & 1.54E-03    & 2.40E-03    \\
$\widehat{\sigma}$     & 9.55E-05     & 4.47E-05     & 1.52E-05  & 5.52E-05    & 1.21E-04    \\
$\widehat{\sigma}^+$   & 1.09E-03     & 7.36E-04     & 1.14E-04  & 6.91E-04    & 1.02E-03    \\ \hline
\end{tabular}
\caption{
Sample means and standard deviations of hedging errors $\eps^\pm$ and $\eps$ for the volatility derivative with payoff \eqref{eq:phi.exp}.
}
\label{tbl:stats-exp-pos}
\end{table}


\begin{table}[ht]
\centering
\begin{tabular}{c|ccccc}
\hline
                       & $\rho=-0.99$ & $\rho=-0.66$ & $\rho=0$  & $\rho=0.66$ & $\rho=0.99$ \\ \hline
$\Re \widehat{\eps}^-$ & 1.52E-03      & 9.76E-04  & -1.01E-06 & -7.98E-04 & -1.11E-03       \\
$\widehat{\eps}$   		 & 2.66E-04      & 1.20E-04  & -1.31E-06 & 1.27E-04  & 2.90E-04        \\
$\Re \widehat{\eps}^+$ & 1.52E-03      & 9.76E-04  & -1.01E-06 & -7.98E-04 & -1.11E-03       \\ \hline
$\Re \widehat{\sigma}^-$   & 5.09E-04      & 3.36E-04  & 5.39E-05  & 2.83E-04  & 4.01E-04        \\
$\widehat{\sigma}$     & 8.53E-05      & 4.20E-05  & 1.34E-05  & 3.66E-05  & 8.22E-05        \\
$\Re \widehat{\sigma}^+$   & 5.09E-04      & 3.36E-04  & 5.39E-05  & 2.83E-04  & 4.01E-04        \\ \hline
\end{tabular}
\caption{
Sample means and standard deviations of hedging errors $\Re \eps^\pm$ and $\eps$ for the volatility derivative with payoff \eqref{eq:phi.exp.2}.
}
\label{tbl:stats-exp-neg}
\end{table}


\begin{table}[ht]
\centering
\begin{tabular}{c|ccccc}
\hline
                       & $\rho=-0.99$ & $\rho=-0.66$ & $\rho=0$  & $\rho=0.66$ & $\rho=0.99$ \\ \hline
$\Re \widehat{\eps}^-$ & 3.54E-03      & 1.33E-03  & -3.08E-05 & 4.98E-04 & 2.26E-03         \\
$\widehat{\eps}$   & 2.95E-03      & 9.57E-04  & -3.09E-05 & 8.70E-04 & 2.86E-03         \\
$\Re \widehat{\eps}^+$ & 3.54E-03      & 1.33E-03  & -3.08E-05 & 4.98E-04 & 2.26E-03         \\ \hline
$\Re \widehat{\sigma}^-$   & 1.51E-03      & 9.44E-04  & 2.33E-05  & 9.24E-04 & 1.63E-03         \\
$\widehat{\sigma}$     & 1.47E-03      & 9.02E-04  & 1.35E-05  & 9.71E-04 & 1.68E-03         \\
$\Re \widehat{\sigma}^+$   & 1.51E-03      & 9.44E-04  & 2.33E-05  & 9.24E-04 & 1.63E-03         \\ \hline
\end{tabular}
\caption{
Sample means and standard deviations of hedging errors $\Re \eps^\pm$ and $\eps$ for the volatility derivative with payoff \eqref{eq:put.payoff}.
}
\label{tbl:stats-put}
\end{table}


\begin{table}[ht]
\centering
\begin{tabular}{c|ccccc}
\hline
                       & $\rho=-0.99$ & $\rho=-0.66$ & $\rho=0$  & $\rho=0.66$ & $\rho=0.99$ \\ \hline
$\Re \widehat{\eps}^-$ & -2.13E-03     & -1.95E-03 & 5.79E-05 & 2.75E-03 & 4.82E-03         \\
$\widehat{\eps}$   & 1.26E-03      & 3.49E-04  & 5.92E-05 & 3.44E-04 & 1.18E-03         \\
$\Re \widehat{\eps}^+$ & -2.13E-03     & -1.95E-03 & 5.79E-05 & 2.75E-03 & 4.82E-03         \\ \hline
$\Re \widehat{\sigma}^-$   & 9.47E-04      & 5.18E-04  & 1.38E-04 & 6.53E-04 & 1.03E-03         \\
$\widehat{\sigma}$     & 5.93E-04      & 2.01E-04  & 3.73E-05 & 2.17E-04 & 7.03E-04         \\
$\Re \widehat{\sigma}^+$   & 9.47E-04      & 5.18E-04  & 1.38E-04 & 6.53E-04 & 1.03E-03         \\ \hline
\end{tabular}
\caption{
Sample means and standard deviations of hedging errors $\Re \eps^\pm$ and $\eps$ for the volatility derivative with payoff \eqref{eq:swap.payoff}.
}
\label{tbl:stats-swap}
\end{table}

\end{document}